\documentclass{article}
\usepackage{graphicx} 
\usepackage{amsfonts,amssymb,graphicx,amsmath,amsthm, float}
\usepackage{comment,xcolor}
\usepackage{epstopdf}
\usepackage[top=1in, bottom=1in, left=1.5in, right=1.5in]{geometry}

\usepackage[colorinlistoftodos]{todonotes}
\theoremstyle{plain}
\newtheorem{theorem}{Theorem}[section]

\theoremstyle{definition}
\newtheorem{example}[theorem]{Example}

\title{An Evaluation of Borda Count Variations Using Ranked Choice Voting Data}
\author{N. Bradley Fox and Benjamin Bruyns}
\date{\today}

\begin{document}

\maketitle






\begin{abstract}
The standard voting methods in the United States, plurality and ranked choice (or instant runoff) voting, are susceptible to significant voting failures.  These flaws include Condorcet and majority failures as well as monotonicity and no-show paradoxes.  We investigate alternative ranked choice voting systems using variations of the points-based Borda count which avoid monotonicity paradoxes.  These variations are based on the way partial ballots are counted and on extending the values of the points assigned to each rank in the ballot.  In particular, we demonstrate which voting failures are possible for each variation and then empirically study 421 U.S. ranked choice elections conducted from 2004 to 2023 to determine the frequency of voting failures when using five Borda variations.  Our analysis demonstrates that the primary vulnerability of majority failures is rare or nonexistent depending on the variation.  Other voting failures such as truncation or compromise failures occur more frequently compared to instant runoff voting as a trade-off for avoiding monotonicity paradoxes.

\end{abstract}


\section{Introduction}

The vast majority of elections in the United States have been conducted using the single-vote plurality method, other than the use of the Electoral College for Presidential elections.  Increasingly over the past twenty years, ranked choice voting (RCV), particularly through the method of instant runoff voting (IRV), has been implemented across numerous municipal and statewide elections including Presidential primaries and U.S. Representative elections in states such as Alaska and Maine.  Both of these methods are vulnerable to multiple voting failures. Plurality is especially problematic, since ranked ballots are necessary to determine if some failures, such as the majority loser or Condorcet failures, have occurred. The following are voting failures, which will be defined in the next section, that can occur with these two methods:
\begin{itemize}
    \item Both methods can exhibit a Condorcet winner failure.
    \item Both methods are susceptible to the spoiler effect.
    \item Both methods are vulnerable to compromise failures.
    \item Plurality can result in the majority loser or the Condorcet loser candidate winning the election.
    \item IRV can exhibit upward and downward monotonicity paradoxes.
    \item IRV can have a truncation failure through occurrences of the no-show paradox.
\end{itemize}

Recent work has focused on comparing results of various voting methods using U.S. elections conducted with IRV. McCune and McCune (2024) showed that each pair of methods agreed on a winning candidate in the vast majority of cases with the lowest pair being the plurality method compared to a variation of the Borda count which produced the same winner about $76\%$ of the time. Graham-Squire and McCune~(2022) examined a similar set of RCV elections that used the IRV method to determine the rates at which voting failures occur, finding that while uncommon, IRV elections sometimes exhibit the various failures listed above. Arrow~(1963) famously proved in his Impossibility Theorem that no voting method exists that will always avoid these failures, but we can still strive to limit their occurrences and particularly eliminate the most undesirable paradoxes. 

To this end, we will examine variations of the points-based Borda count method.  This voting procedure was first introduced by Nicholas of Cusa in the 15th century, although later was attributed to 18th century mathematician Jean-Charles de Borda~(Emerson 2013).  It is an RCV method in which each voter ranks the candidates from 1 to $n$, with their top choice receiving $n$ points, second earning $n-1$, and so on until their last choice gets $1$ point.  We'll represent the points awarded to the $n$ candidates by $(n,n-1,n-2,\ldots, 2,1)$, which will be referred to as the points vector.  The candidate receiving the highest point total after accumulating points from every voter is declared the winner.

The advantages of the Borda count over IRV have been previously described by authors such as Dummett~(1997).  In his \textit{Principles of Election Reform}, he argues against the single transferable vote, which is a multiwinner version of IRV, stating ``STV is quasi-chaotic because it takes into account only the first choices of some voters, and the second, third, or even fourth choices of others, giving them as much weight as the first choices."  The Borda count, on the other hand, is posited by Saari~(2001a) to be ``the more accurate reflection of voter preferences" since it considers the full set of rankings and weighs them according to the voter's ordering.  He further argues that the Borda count ``is the unique method where the outcome avoids most of the unpleasant paradoxes"~(Saari 2001b) described by Arrow.

The use of RCV has increased in the U.S., but this has primarily been through IRV elections.  Variations of the Borda count have been successfully implemented internationally, in particular in the Pacific Island nations of Kiribati and Nauru~(Reilly 2022). Kiribati uses the standard Borda count in which members of their legislative body rank four members to be nominees for their presidency.  The top four points recipients then move on to a national plurality election.  Nauru uses a fractional version called the Dowdall system to elect representatives from multi-member districts. In their Borda variation, the highest ranked candidate receives 1 point, second receives half a point, third earns a third of a point, and so on for a points vector of $(1,1/2, 1/3,1/4,\ldots, 1/n)$. In addition to these island nations, Slovenia uses the Borda count to elect two parliamentarian members who represent Italian and Hungarian minority districts~(Fraenkel and Grofman 2014).

While Borda methods have not been used in U.S. political elections, Americans have seen Borda count variations used in areas such as sports rankings and award voting (Morris and Mantrell 2024).  For example, the NCAA uses truncated 25-team ballots to determine their weekly Top 25 collegiate sport rankings.  Similarly, MLB uses a version of Borda for their MVP and other award voting with a points scale that favors first-place votes.  Particularly, a ten-player ballot awards 14 points for their top choice before following the standard Borda points scale from 2nd to last, resulting in assigning points as $(14,9,8,\ldots, 2,1)$. The variations of how countries or organizations set up the points vector, or whether ballots can be truncated or partial, adds to the flexibility of how states or cities could implement their own versions of the Borda count.

We will discuss numerous voting failures in this article, analyzing which are potentially exhibited by variations of the Borda count.  Specific U.S. RCV elections will be examined in detail, showing some vulnerabilities of these points-based methods and comparing to similar failures that can occur with IRV.  Finally, we empirically study the voting failure rates if Borda count methods were used in the over 400 U.S. RCV elections from the past twenty years.

\section{Voting Failures}

We begin by defining the primary means of evaluating the fairness of election methods, which is through the observation of voting failures, sometimes referred to in the positive sense as fairness criteria.  It is known from Arrow's Impossibility Theorem that every voting method can exhibit at least one of a small subset of these failures. We will demonstrate which failures are possible or will never occur when using the Borda count variations, which will be aided by grouping the failures into three distinct sets.

\begin{itemize}
    \item Verifiable Failures: These voting failures occur when a candidate meets a criterion that deems them the fair winner of the election, yet the method does not result in their election.  Or similarly, a criterion may declare a candidate to be an unfair choice as the winner, yet the method does in fact elect that candidate.  In each case, it is easy to verify from a particular election's preference profile, which is a table displaying the number of votes for each unique ballot ranking, if such a criterion is met and whether that candidate won or lost. 
        \begin{itemize}
            \item \textbf{Majority Winner/Loser Failure:} A \textit{majority winner failure} occurs when a candidate has a majority of first-place votes but loses the election. A \textit{majority loser failure}, on the other hand, occurs when a candidate has a majority of last-place votes yet wins the election. Note that partial ballots with more than one candidate unranked will not be considered as having a last-place vote, but are still considered in the total when determining if a majority exists.
            \item \textbf{Condorcet Winner/Loser Failure:} A \textit{Condorcet winner failure} occurs if a candidate would defeat every other candidate in a pairwise comparison, but loses the election. A \textit{Condorcet loser failure} occurs if a candidate loses to every other candidate in a head-to-head comparison, yet wins the election. With partial ballots, note that a candidate ranked on the partial ballot is clearly favored over all unranked candidates, but partial ballots are not considered in a pairwise comparison if both candidates are unranked on that ballot. 
        \end{itemize}
        \item Antidemocratic Voting Paradoxes: These are voting failures where changes in votes paradoxically have an inverse effect than expected on the winner of the election.  Such a paradox is antithetical to the democratic process since gaining votes or moving up the ranking of ballots should only possibly help the candidate and never hurt their chances of winning, and vice versa for losing support.
        \begin{itemize}
    \item \textbf{Upward/Downward Monotonicity Paradox}:  An \textit{upward monotonicity paradox} occurs if there is a set of ballots in which moving the winning candidate higher in those rankings, while maintaining the relative order of the other candidates, would result in that candidate no longer winning the election.  A \textit{downward monotonicity paradox}, on the other hand, is exhibited if there is a set of ballots in which moving a losing candidate down in those rankings, while maintaining the relative order of the other candidates, would result in that candidate now winning the election.  These are also referred to as the more-is-less (upward) and less-is-more (downward) paradoxes.
    \item \textbf{No-Show Paradox}: The \textit{no-show paradox} is an extreme form of the soon-to-be defined truncation failure that occurs if there is a set of voters whose top preference is a losing candidate, but removing those ballots would result in that candidate now winning the election.  Hence those voters have a more desirable outcome occur by staying home and not voting.
\end{itemize}
    \item Voter/Candidate Manipulation Failures: This set of voting failures involves changes to voter ballots or to the slate of candidates that could alter the election results.  While these could make a voting method vulnerable to types of tactical or strategic voting, it may be difficult in practice for voters to collectively manipulate the election.  And unlike the previous paradoxes, the changes made to votes do have the expected effect, such as moving a candidate down or truncating them off a ballot resulting in that candidate losing.  
        \begin{itemize}
            \item \textbf{Truncation Failure:} A \textit{truncation failure} occurs if there is a set of voters who could rank fewer candidates on their ballots and result in their favored candidate becoming the winner. Note this is usually called a truncation paradox in the literature, but we choose to distinguish it from the antidemocratic paradoxes since truncating to help your favored candidate is having the desired effect for that voter rather than the inverse effect.
            \item \textbf{Compromise Failure:} Assume there is a set of voters that rank a candidate $A$ over candidate $B$, both of whom are ranked higher than the winning candidate $C$. A \textit{compromise failure} occurs if those voters moving candidate $B$ to the top of their ballot would result in the winner changing to $B$.
            \item \textbf{Spoiler Effect:} A \textit{spoiler effect} occurs if the removal of a subset of losing candidates from the slate would change the winner of election.
        \end{itemize}
    
\end{itemize}


\begin{table}
{\scriptsize
\begin{center}
\begin{tabular}{ | c | c | c | c | c | c | c | c | c | c | }
\hline
 $\#$Votes & 11181 & 27165 & 15488 & 21177 & 34155 & 3678 & 23650 & 47407 & 4699 \\ \hline
 1st & Begich & Begich & Begich & Palin & Palin & Palin & Peltola & Peltola & Peltola \\
 2nd & -- & Palin & Peltola & -- & Begich & Peltola & -- & Begich & Palin \\
 3rd & -- & Peltola & Palin & -- & Peltola & Begich & -- & Palin & Begich \\ \hline
\end{tabular}
\caption{The Alaska Special Election for the US House in August 2022.  Write-ins are removed and two-candidate ballots have the unlisted candidate included in third rank.}
\label{Alaska}
\end{center}}
\end{table}

Analysis by Graham-Squire and McCune (2022) showed that IRV has exhibited the three antidemocratic paradoxes in multiples elections over the past twenty years. The most notable example, further analyzed by Graham-Squire and McCune (2023), is the \textbf{2022 Alaska U.S. House Special Election}, which featured upward monotonicity and no-show paradoxes.  The preference profile for this election is displayed in Table~\ref{Alaska}.  
The election was conducted using IRV, resulting in Peltola as the winner, and we note that Peltola is also the plurality winner, although without a majority of first-place votes. 

In total, this IRV election exhibits five voting failures, including the upward monotonicity paradox that can result in Peltola gaining votes but losing to Begich and the no-show paradox where some voters who prefer Begich over Peltola could cause Begich to win by not voting.  There were also a Condorcet winner failure since Begich is the Condorcet winner and a spoiler effect since removing Palin results in Begich winning. 
 Finally, there is a compromise failure in which Palin supporters could compromise and move their second choice Begich up to first, resulting in Begich winning.

Focusing more on the monotonicity paradox, it would occur if 6000 voters who included only Palin on their ballot changed their vote to rank Peltola above Palin.  This would result in Palin being eliminated first and Begich now winning the final round over Peltola.  Graham-Squire and McCune~(2022) dismiss concerns about monotonicity paradoxes by arguing that ``[i]t is hard to believe that voters would vote insincerely to attempt to engineer a monotonicity or no-show paradox."  

Our contention is that antidemocratic voting paradoxes could present themselves in IRV elections not through intentional strategic voting, but through the standard ways that votes change during campaigns. New policy positions, endorsements, debate performances, negative ads, and scandals would be more likely to trigger a voting paradox than strategic voting. Those 6000 voters could have simply been persuaded in the final moments of the campaign to vote for Peltola and accidentally led to her downfall. 

Moreover, since elections are simply snapshots in time of a set of voter preferences on election day, it is impossible to know if last-minute changes to preferences or to turnout led to a monotonicity or no-show paradox.  The upward monotonicity paradox is merely a hypothetical scenario that clearly did not happen in this election, but the reverse process involving a downward monotonicity paradox may have actually occurred without us knowing.  In the week leading up to election day, there could have been 6000 Peltola$>$Palin voters (meaning ones who ranked Peltola above Palin) who were convinced to drop Peltola from their ballots and only vote for Palin, resulting in the final ballot totals in Table~\ref{Alaska}.  Now we have a scenario where Begich would have been the winner, but Peltola losing support caused her to become the new winner, creating a downward monotonicity paradox.  

While Graham-Squire and McCune~(2022) showed that it is rare for a monotonicity paradox to be possible, the fact that we cannot pinpoint whether they actually occurred or not is a clear weakness of IRV. This fact motivates our focus on \textit{monotonic voting methods} which can never display one of these three paradoxes. Our focus will be on variations of the Borda count, but other such voting systems used in U.S. or international elections are approval voting and Bucklin voting, along with the plurality method.  Additionally, there are multiple methods based on pairwise comparisons such as Copeland's method, the Ranked Pairs system, and the Schulze method, although we are not aware of their use in any political elections.

\section{Borda Count Variations}

We will examine the outcomes and fairness of several variations of the Borda count method in our upcoming sections, which we will refer to as BC variations. We must first distinguish between two implementations of the standard Borda count points system described by Emerson~(2013) and used by McCune and McCune~(2024) to account for partial or truncated ballots. The original points system introduced by Borda with the points vector $(n,n-1,\ldots, 2, 1)$ is assuming every ballot is complete; that is, each voter ranked a full list of candidates.  In most ranked choice elections with $n$ candidates, however, voters may choose to submit a \textit{partial} ballot in which they rank fewer than $n$ options.  Additionally, some districts limit voters to ranking only a subset of candidates of a maximum size $m<n$.  These are called \textit{truncated} ballots.  For example, the Minneapolis City Council elections only allow at most 3 candidates to be included in a ballot no matter how many candidates are running.

The two primary approaches for handling partial or truncated ballots are the \textit{Borda count untreated} (BCU) and the \textit{modified Borda count} (MBC).  If a ballot only ranks $k$ of the $n$ candidates, MBC awards $(k,k-1,\ldots, 2,1,0,\ldots,0)$, so all unranked candidates receive 0 points.  BCU awards points to the ranked candidates based on the full Borda count value as $(n,n-1,\ldots, n-k+1,0,\ldots, 0)$. Emerson~(2013) argues that BCU would encourage voters to use a form of tactical voting in which voters cast a partial ballot consisting of only their top choice, which is known as \textit{bullet voting}.  Giving a full $n$ points to that choice and 0 points to all other candidates would clearly give a voter's top candidate the best chances of winning, and if enough voters follow this strategy, the election would devolve into a plurality vote and additionally lead to frequent truncation failures. 

One negative of both variations is that not all votes are considered equal. For example, in an election with $5$ candidates, a voter who fills out a complete ballot ranking all 5 candidates is assigning 15 total points to the candidates.  A voter with a partial ballot of 3 candidates would only give candidates a total of 12 points using BCU and merely 6 points using MBC. While this does incentivize the electorate to fill out complete ballots, it is against democratic principles to count some voters more than others.  Furthermore, in elections such as the 2021 East Hampton Mayor election, as many as 66 candidates may be present on the ballot.  In elections with more than 8 candidates, which account for over 10 percent of the RCV elections in our dataset, it would not feasible to expect each voter to differentiate every candidate enough to provide a complete ranking without arbitrarily ordering the lower end of candidates. MBC in particular penalizes voters who want to submit a partial ballot to the point that they may be pressured to vote insincerely to ensure their top candidate is awarded the maximum number of points.

To address this inequity, we favor the variation of the Borda count introduced by Dummett~(1997) called the \textit{averaged Borda count} (ABC).  In the previous example with 5 candidates and a 3-candidate partial ballot, it initially follows BCU with 5 points for 1st, 4 for 2nd, and 3 for 3rd, but ABC will then average the unassigned 3 points to give 1.5 points to each of the unlisted candidates.  In general, if $k$ candidates are ranked in an $n$-candidate election, they receive the same points as in BCU, while the unranked $n-k$ candidates receive an equal share of the remaining points that a complete ballot would have assigned.  This equates to $\displaystyle \frac{1}{n-k}(1+2+\cdots +(n-k+1)+(n-k))=\frac{(n-k)(n-k+1)}{2(n-k)}=\frac{n-k+1}{2}$.  In total, the points vector~is 
$$\left(n,n-1,n-2,\ldots, n-k+1,\frac{n-k+1}{2},\ldots ,\frac{n-k+1}{2}\right).$$ 
Observe that this method assigns the same total of points for every ballot, whether complete or partial. This also reduces, albeit not completely eliminating, the bullet voting incentive of BCU. A voter in an ABC election who truncates their ballot is making a net neutral move, as their highest ranked candidate that is truncated will lose points, but one or more lowest ranked candidates will gain points. Hence we expect fewer truncation failures from ABC compared BCU. This aligns with theoretical work by Kamwa~(2022), who also showed that MBC cannot exhibit truncation failures.

We will also focus on variations of the Borda count in which the points scale is extended, similar to its use in Nauru with the Dowdall system or in sports award voting. The motivation behind this approach is twofold. One is that the importance of a voter's choice between their top two candidates exceeds that of their choice between subsequent pairs of consecutively ranked candidates.  Hence the relative difference in points should be larger toward the top of the ballot.  Secondly, a larger difference in points, specifically between the top two candidates, would lessen the probability of a majority failure compared to the standard Borda count, which is the most frequent criticism of Borda.  The trade-off in fairness, however, is a possible increase in voter manipulation failures.

For creating such a variation for an election with $n$ candidates, let $p_i$ be the number of points assigned to the $i$th ranked candidate. The standard Borda count uses $p_i=n-i+1$ to assign the $(n,n-1,\ldots, 1)$ points. There are infinitely many point scales that could be used with the assumption that a higher ranked candidate earns at least as many as a lower ranked candidate ($p_i\geq p_{i+1}$), and that the difference between pairs of consecutively ranked candidates is nonincreasing as you move down the ballot ($p_i-p_{i+1}\geq p_{i+1}-p_{i+2})$.  Note that strict inequalities cannot be used on those assumptions, else we would be excluding truncated ballots for the first assumption and likewise the standard Borda count's use of equal differences for the second assumption.  When using a sufficiently large difference $p_1-p_2$, the variation would tend toward a plurality election.  

The two points-scale variations that we introduce are the \textit{Exponential Borda Count} (EBC) and \textit{Quadratic Borda Count} (QBC).  EBC awards 1 point to the last ranked candidate and double the points for each subsequently higher rank.  This leads to the points for the $i$th candidate being the exponential function $p_i=2^{n - i}$, or points vector of $(2^{n-1}, 2^{n-2},\ldots, 4,2,1)$.  QBC also awards 1 point for a last-place vote and then the difference in points between ranks increases by 1 each time with the top two candidates having a difference of $p_1-p_2=n-1$.  In general, $p_{i}-p_{i+1}=p_{i+1}-p_{i+2}+1$, and the formula for the points assigned to the $i$th ranked candidate is the quadratic function
$$p_i=1+\sum_{j=1}^{n-i} j=1+\frac{(n-i)(n-i+1)}{2}.$$
The resulting points vector is 
$$\left(1+\frac{(n-1)n}{2}, 1+\frac{(n-2)(n-1)}{2},\ldots,7,4,2,1\right). $$
Table~\ref{ballotpoints} shows the values of $p_i$ in a 6-candidate election to illustrate the differences in the variations. 

\begin{table}
{\begin{center}
\begin{tabular}{ | c | c | c | c | }
\hline
 Rank & EBC pts & QBC pts & ABC/BCU/MBC pts \\ \hline
 1st & 32 & 16 & 6 \\ \hline
 2nd & 16 & 11 & 5 \\ \hline
 3rd & 8 & 7 & 4 \\ \hline
 4th & 4 & 4 & 3  \\ \hline
 5th & 2 & 2 & 2 \\ \hline
 6th & 1 & 1 & 1 \\ \hline
\end{tabular}
\caption{The points $p_i$ that each BC variation assigns to a $i$th ranked candidate in a complete ballot with 6 candidates. Note that ABC, BCU, and MBC only differ on partial ballots.}
\label{ballotpoints}
\end{center}}
\end{table}

Notice that each variation we've defined assigns 1 point to the last-place candidate.  The Dowdall system, which is currently used in elections in Nauru, uses fractional points with $1/n$ assigned for a last-place vote if there are $n$ candidates.  
This system could be viewed as a similar extension of the Borda count if we convert each fraction to an integer by multiplying each point value by the least common multiple of the numbers $1, 2, \ldots, n$. Then we subtract the new points by $\text{lcm}\{1,2,\ldots, n\}/n-1$ to make $p_n=1$.  
Note that these two operations applied to every rank would not change the order of the candidates with regard to their final point totals. In the case of $n=6$ candidates, the points vector is transformed from $(1,1/2,1/3,1/4,1/5,1/6)$ to $(51, 21, 11, 6, 3, 1)$.  This resembles the rate of increase for EBC, yet it does not increase by the same factor from $p_{i+1}$ to $p_i$ for each $i$, with the factors varying between just under 2 and 3.  Because of this irregularity, we choose to omit this variation from our analysis.

In order to maintain equity between ballots, we will apply the averaged Borda count process regarding partial ballots for our EBC and QBC variations, as opposed to having three versions of each. As an example, using EBC in an election with six candidates in which two of the candidates are ranked, the remaining points would be 15, making a points vector of $(32, 16, 3.75, 3.75, 3.75, 3.75)$.




Now that we have defined the variations, we shift our focus to whether each voting failure can occur with each method.  Table~\ref{FailureChart} displays this information for the currently used plurality and IRV methods, along with the five BC variations we are investigating.  Any failure that is possibly exhibited by these variations will either occur within the empirical analysis in an upcoming section, or is demonstrated by the referenced example.  The failures for which a Borda variation is immune will be proven so in the upcoming theorems, with the exception of MBC avoiding truncation failures, as this was previously shown (Kamwa 2022).

\begin{table}
{
\begin{center}
\begin{tabular}{ | c | c | c | c | c | c | c | c |}
\hline
 Voting Failure & Plur. & IRV & EBC & QBC & ABC & BCU & MBC \\ \hline
 Majority Winner        & N & N & Y(\ref{MajEBC-QBC}) & Y(\ref{MajEBC-QBC}) & Y & Y & Y \\
 Majority Loser         & Y & N & Y(\ref{CondMajLosEBC-QBC}) & Y(\ref{CondMajLosEBC-QBC}) & N(\ref{CondLosABC})  & Y(\ref{CondMajLosBCU}) & N(\ref{MajLosMBC}) \\
 Condorcet Winner       & Y & Y & Y & Y & Y & Y & Y \\
Condorcet Loser        & Y & N & Y(\ref{CondMajLosEBC-QBC}) & Y(\ref{CondMajLosEBC-QBC}) & N(\ref{CondLosABC}) & Y(\ref{CondMajLosBCU}) & Y(\ref{CondLosMBC}) \\ \hline
 Upward Mono.    & N & Y & N & N & N & N & N \\
 Downward Mono.  & N & Y & N & N & N & N & N \\
 No-Show                & N & Y & N & N & N & N & N \\ \hline
 Truncation             & N &N$^*$& Y & Y & Y & Y & N \\
 Compromise             & Y & Y & Y & Y & Y & Y & Y \\
 Spoiler Effect         & Y & Y & Y & Y & Y & Y & Y \\ \hline 
\end{tabular}
\caption{Voting failures that can be exhibited (Y) or cannot occur (N) for each voting methods, with N$^*$ listed for Truncation-IRV since No-Show is the only form of a truncation failure that can occur.}
\label{FailureChart}
\end{center}}
\end{table}

The following are elections of very small voter size that demonstrate particular voting failures for some of our BC variations.  Although these are toy examples to make the points calculation and pairwise comparisons easy, examples can be developed for much larger elections by keeping similar distributions in ballot types. As we will see in Section~\ref{sect:results}, none of these voting failure types would have occurred using the specified Borda variation on previous U.S. elections, hence the need for examples to demonstrate their possibility.  We begin by demonstrating that EBC and QBC can exhibit a majority winner failure. 

\begin{example}\label{MajEBC-QBC}
    Consider an election with candidates $A$, $B$, and $C$ using EBC, which note is equivalent to QBC when there are only 3 candidates since both would use the points vector $(4,2,1)$. Suppose there are 4 ballots ordering the candidates $A>B>C$ and $3$ with the ranking $B>C>A$.  Candidate $A$ clearly has a majority of first-place votes, yet $A$ only receives 19 points compared to candidate $B$'s 20 points, creating a majority winner failure.
\end{example}

Note the previous example is also a Condorcet winner failure, but we only focused on the majority aspect because it does not occur in our upcoming empirical analysis, whereas Condorcet winner failures are exhibited for every variation of the Borda count.  The next example shows that EBC and QBC can both have majority loser and Condorcet loser failures.  Although the Borda count cannot exhibit these two failures, as shown by Fishburn and Gehrlein~(1976), any extension of Borda with a points vector where the difference $p_i-p_{i+1}$ increases as $i$ goes from $1$ to $n-1$ can have these failures if the number of voters is sufficiently large.

\begin{example}\label{CondMajLosEBC-QBC}
    Consider an election with candidates $A$, $B$, and $C$ using EBC or the equivalent QBC.  Assume there are 4 ballots ranking $A>B>C$, 3 ranking $A>C>B$, 4 with $B>C>A$, and 4 with $C>B>A$.  Notice that $A$ is ranked last in 8 of the 15 ballots, which is a majority and leads $A$ to lose in pairwise competition to both $B$ and $C$. Candidate $A$ wins the election, however, by earning 36 points while $B$ earns 35 and $C$ receives 34.  This election exhibits both a majority loser and Condorcet loser failure.
\end{example}

The final two examples relate to Condorcet and majority loser failures for BCU and MBC.  Notice that only Condorcet loser is shown for MBC, as we will prove that it is immune to majority loser failures. 

\begin{example}\label{CondMajLosBCU}
    Consider an election using BCU with candidates $A$, $B$, and $C$ and only five votes: 2 ranking only $A$, 2 ballots with $B>C>A$, and 1 ballot ranking $C>B>A$.  Observe that $A$ is ranked last in a majority of ballots, particularly 3 of the 5, which also implies $B$ and $C$ are both favored head-to-head over $A$.  Yet the total points for $A$ is 9, while $B$ and $C$ earn 8 and 7 points, respectively.  This exhibits both a majority loser and Condorcet loser failure.
\end{example}

\begin{example}\label{CondLosMBC}
    Consider an election using MBC with 3 candidates $A$, $B$, and $C$ and the following ballots: 2 ranking $A>B>C$, 2 with $A>C>B$, 5 partial ballots ranking only $B$, and 5 ranking only $C$.  Candidates $B$ and $C$ are both favored over $A$ by a margin of 5 to 4, implying $A$ is a Condorcet loser candidate.  However, $A$ wins the election with 12 points compared to 11 points for each of $B$ and $C$, resulting in a Condorcet loser failure.
\end{example}

It was shown (Fishburn and Gehrlein 1976) that the Borda count, assuming all ballots are complete, cannot have a Condorcet loser failure.  This implies a majority loser failure is also impossible.  We now show two variations of Borda, particularly MBC and ABC, are also immune to majority loser failures, and in the case of ABC, Condorcet loser failures as well.

\begin{theorem}\label{MajLosMBC}
The modified Borda count (MBC) can never exhibit a majority loser failure.
\end{theorem}
\begin{proof}
     We will show by contradiction that MBC, similar to the Borda count, cannot have a majority loser failure.  Suppose there is an election where MBC does have this failure with candidate $A$ winning despite having a majority of last-place votes. Note that these last-place votes are all from complete ballots, with $A$ being last on over half of all total votes.  There must, however, exist partial ballots, else this election would contradict the Borda count's avoidance of this failure.  

    Consider a partial ballot ranking $k$ of the $n$ candidates as $C_1>C_2>\cdots >C_k$ and leaving the remaining candidates $D_1, D_2,\ldots, D_{n-k}$ unranked.  Note the points vector for this partial ballot is $(k,k-1,\ldots, 1, 0,\ldots,0)$.  
    
    In the case of the majority loser candidate $A$ being unranked, say as the candidate $D_1$, we can complete this ballot as $C_1>C_2>\cdots >C_k>A>D_2>D_3>\cdots >D_{n-k}$.  The new points vector would be the standard Borda vector of $(n,n-1,\ldots, 1)$ where candidate $A$ now receives $n-k$ points after initially receiving 0 points. Observe that all candidates $C_j$ also increased their points by $n-k$, while each $D_j$ for $j\neq 1$ increased by a value less than $n-k$.  Therefore, candidate $A$ added at least as many points compared to all other candidates from completing all partial ballots of this type.

    On the other hand, if candidate $A$ was ranked where $C_i=A$ for some $i=1$ to $k$, then we complete the ballot by adding the unranked candidates $D_1>D_2>\cdots>D_{n-k}$ in that order after $C_k$.  The points vector again becomes the standard $(n,n-1,\ldots, 1)$ with $A$ gaining $n-k$ points, the same as the other candidates $C_j$ and $D_1$.  The lower ranked $D_j$ candidates only increase by at most $n-k-1$ points, so once again, $A$ increases by at least as many points through this ballot completion case as all other candidates.

    We now have a Borda count election without any partial ballots, and since $A$ gained as many or more points than any other candidate, $A$ remains the winner of the election.  Although new last-place votes were introduced, the number of last-place votes for $A$ remained the same, and hence this candidate still has a majority of last-place votes.  Therefore, this new election with all ballots completed has a majority loser failure, which cannot exist with the Borda count, hence a contradiction.  Thus, MBC cannot have such a failure.
\end{proof}

\begin{theorem}\label{CondLosABC}
The averaged Borda count (ABC) can never exhibit a Condorcet loser failure, and thus also cannot have a majority loser failure.
\end{theorem}
\begin{proof}
We follow a similar approach to the previous result by supposing for a contradiction that there is an ABC election with a Condorcet loser failure.  Assume $A$ is the winning candidate who loses in pairwise competition to all other candidates.  The standard Borda count does not have Condorcet loser failures, so there must exist partial ballots.

Consider a partial ballot ranking candidates $C_1>C_2>\cdots >C_k$ with candidates $D_1,D_2,\ldots, D_{n-k}$ unranked.  The points vector using ABC would be $$\left(n,n-1,\ldots, n-k+1, \frac{n-k+1}{2},\frac{n-k+1}{2},\ldots ,\frac{n-k+1}{2}\right).$$

We now consider a new election in which each of the complete ballots are included twice, and we complete the previously partial ballots in two ways.  First, we add in the unranked candidates to form the ballot $C_1>C_2>\cdots >C_k>D_1>D_2>\cdots >D_{n-k-1}>D_{n-k}$.  We also form a complete ballot using the reverse order on the $D_j$ candidates as $C_1>C_2>\cdots >C_k>D_{n-k}>D_{n-k-1}>\cdots >D_{2}>D_1$.  Each option now has the standard $(n,n-1,\ldots , 1)$ points vector.

Focusing on the points received by each candidate in the new election, the originally complete ballots are now included twice, resulting in double the points for each candidate.  Meanwhile, the originally partial ballots assign $n-j+1$ each time to a candidate $C_j$, so this candidate also receives twice as many points in the new election.  Finally, a candidate $D_j$ earns $n-k-j+1$ and $j$ points, respectively, in the two forms of completed ballots.  Together, this is $n-k+j-1+j=n-k+1$, which is exactly twice the $(n-k+1)/2$ points received as an unranked candidate through ABC. Ultimately, every candidate's total points are doubled in the new election with all complete ballots, maintaining the win for candidate A in this Borda count election.

We now focus on the pairwise comparisons in which candidate $A$ lost to every other candidate in the original election.  Let us consider a single candidate $B$ that is favored over $A$ in $x$ ballots with $A$ being favored over $B$ in $y$ ballots.  This includes partial ballots in which one of the candidates is ranked and the other is unranked. Any voter with a partial ballot leaving both $A$ and $B$ unranked does not contribute to either $x$ or $y$.  Since $A$ loses all head-to-head comparisons, we know $x>y$.

In the new Borda count election, all previously complete ballots are included twice. Any ballot where one of $A$ or $B$ was ranked and the other unranked would also count twice toward the pairwise comparison.  This is because the originally ranked candidate is favored in both new ballots, despite the order of the newly ranked candidates being different in the two forms of completing ballots.  This accounts for all cases where one of $A$ or $B$ was ranked higher, and thus there are so far $2x$ ballots with $B$ favored over $A$, and $2y$ with the opposite relative order.  Assuming there were $z$ partial ballots that left $A$ and $B$ unranked, the two forms of ballot completion with reverse ordering result in $A$ ranked higher in $z$ newly completed ballots and $B$ also being higher in $z$ ballots.  Now that we have covered all cases of both being ranked, one ranked and one unranked, and both unranked, $B$ is favored over $A$ in $2x+z$ ballots, while $A$ is favored in $2y+z$ ballots.  The fact that $x>y$ directly implies $2x+z>2y+z$, meaning that $B$ wins the head-to-head comparison.  

We have shown that $A$ wins in this Borda count election with each ballot being complete.  Additionally, we see that $A$ still loses each pairwise comparison, making them the Condorcet loser candidate.  Hence we have a Condorcet loser failure for the Borda count, which is a contradiction.  Thus, it is impossible for ABC to have a Condorcet loser failure.  Since a majority loser candidate must also be a Condorcet loser, the result holds for majority loser failures as well.
\end{proof}

\section{Data and Methodology}

\begin{table}
{\begin{center}
\begin{tabular}{ | c | c | c |  }
\hline
 Location & Num. of Elections & Years \\ \hline
 Alaska & 57 & 2020, 2022 \\ \hline
 Albany, CA & 2 & 2022 \\ \hline
 Berkeley, CA & 28 & 2010--2020 \\ \hline
 Bloomington, MN & 2 & 2023   \\ \hline
 Boulder, CO & 1 & 2023  \\ \hline
 Burlington, VT & 5 & 2006, 2009, 2023 \\ \hline
 Cambridge, MA & 2 & 2023 \\ \hline
 Corvallis, OR & 2 & 2022 \\ \hline
 Easthampton, MA & 1 & 2021 \\ \hline
 Eastpointe, MI & 1 & 2019 \\ \hline
 Elk Ridge, UT & 1 & 2021 \\ \hline
 Hawaii & 1 & 2020 \\ \hline
 Kansas & 1 & 2020 \\ \hline
 Las Cruces, NM & 3 & 2019 \\ \hline
 Maine & 10 & 2018, 2020, 2022 \\ \hline
 Minneapolis, MN & 73 & 2009--2023 \\ \hline
 Minnetonka, MN & 6 & 2021, 2023 \\ \hline
 New York City, NY & 64 & 2021 \\ \hline
 Oakland, CA & 56 & 2010--2022 \\ \hline
 Payson, UT & 1 & 2019 \\ \hline
 Pierce County, WA & 4 & 2008, 2009 \\ \hline
 Portland, ME & 4 & 2021, 2022 \\ \hline
 San Francisco, CA & 64 & 2004--2022 \\ \hline
 San Leandro, CA & 17 & 2010--2020 \\ \hline
 Santa Fe, NM & 3 & 2018 \\ \hline
 Springville, UT & 2 & 2021 \\ \hline
 St. Louis Park, MN & 2 & 2021 \\ \hline
 Takoma Park, MD & 4 & 2022 \\ \hline
 Vineyard, UT & 2 & 2019, 2021 \\ \hline
 Woodland Hills, UT & 1 & 2021 \\ \hline
 Wyoming & 1 & 2020 \\ \hline
\end{tabular}
\caption{List of 421 RCV elections within our dataset, ranging from local elections to U.S. Presidential primaries}
\label{ElectionData}
\end{center}}
\end{table}

The election data used in our empirical analysis consists of 421 RCV elections from a publicly available database~(Otis 2022). Table~\ref{ElectionData} summarizes the dataset, including the states and cities where the elections occurred and the range of years for each jurisdiction. The elections include local city council, school board, and mayor offices in addition to statewide elections for the U.S. House of Representatives and Presidential primaries. 

In using the RCV data where voters chose their ranking based on IRV deciding the winner, we recognize a potential methodological flaw by which some voter preferences may differ if a points-based method was used.  Strategic voting, such as truncating a ballot to remove your favored candidate's top competitor, would nevertheless have little effect on potential verifiable failures.  Meanwhile, the purpose of voter manipulation failures is to pinpoint the potential for this exact behavior and identify it as an unfair outcome, erasing any methodological concerns.




There are multiple considerations needed to process ballots and implement our points-based methods. For instance, there are many ballots in which a candidate was listed multiple times, ranks were skipped, or multiple candidates were listed in the same ranking, identified as an overvote. Any appearance of a candidate after the first instance was removed from the ballot, likewise for skipped or undeclared ranks.  Then these were processed as partial ballots.  Finally, for cases of overvotes, there was no way to determine which candidates were in the overvote, so all candidates below that ranking were removed, leaving only the candidates ranked above the overvote rank.

Another determination that must be made when implementing a points-based voting system is how to enumerate ballots that include write-in candidates.  Across all elections in which write-in candidates were allowed to be included, they received no more than $1.6\%$ of the total first-place votes.  Moreover, this number includes all write-in candidates collectively since data from elections offices usually reports these as a single candidate named Write-In or a similar designation.
Therefore, we opted for removing all write-in candidates since they had a negligible effect on the election outcome.

The resulting ballots were then aggregated in preference profiles, where Python code was used to run each BC variation to determine the winner. For verifiable failures, it is algorithmically straightforward to determine if a majority and Condorcet winner or loser exist by counting first-place votes and running each head-to-head comparison.

The presence of voter/candidate manipulation failures are more difficult to identify from a computational standpoint. For a truncation failure, our algorithm considers each losing candidate and truncates a set of ballots to give them an optimal chance of overtaking the winning candidate's point total. This is accomplished by removing all ranks starting with the initial winner from each ballot where the particular losing candidate was ranked higher. If that candidate now wins, there is a truncation failure.  There are scenarios, however, where a third candidate would now win, forcing us to further analyze those examples to determine if a smaller subset of ballots or if a different truncation strategy, such as removing this third candidate as well, would trigger a failure.  Compromise failures follow a similar algorithmic process in which a losing candidate ranked higher than the winner is moved to the top of all such ballots to determine if they would now win.

To identify spoiler effects, we use a brute-force approach for elections with slates of 10 or fewer candidates.  By sequentially removing each subset of losing candidates, the algorithm determines if the winner has changed, resulting in a spoiler effect. Unfortunately, it becomes computationally infeasible to perform this evaluation for every subset if there is a larger slate. To guarantee accuracy, we limited our analysis to the 396 elections with 10 or fewer candidates.

\section{Analysis of Selected Elections}

 We begin a discussion of whether voting failures occurred in particular RCV elections using the BC variations with the previously mentioned \textbf{2022 Alaska US House Special Election}, as seen in Table~\ref{Alaska}. A partial ballot of two candidates is displayed with the unlisted candidate in the third ranking to simplify the presentation of the profile, although note that we do consider those ballots as being partial when applying MBC and BCU.

Table~\ref{AKResults} displays the point totals for each of the five variations. Compared to the IRV outcome, the BC methods have much more fair results for this election.  First, recall that monotonicity and no-show paradoxes cannot occur with these points-based methods as they do with IRV.  We also observe that Begich receives the most points with any method, which means there was no Condorcet winner failure, implying no verifiable failures are observed.  

\begin{table}
{\begin{center}
\begin{tabular}{ | c | c | c | c | c |  }
\hline
 Candidate & QBC/EBC pts & ABC pts & BCU pts & MBC pts\\ \hline
 Begich & 454203.5  &  400369.5 & 331011 &  233616\\ \hline
 Peltola & 451465 &  375646 & 284939  & 202658 \\ \hline
 Palin & 414972.5 & 355962.5 & 270853 & 184211 \\ \hline
\end{tabular}
\caption{The results from the five BC variations on the Alaska Special Election for the US House in August 2022. With three candidates, a full ballot would earn $(3,2,1)$ points for ABC/MBC/BCU and $(4,2,1)$ for each of QBC and EBC.}
\end{center}}
\label{AKResults}
\end{table}

Since Begich wins each head-to-head matchup between candidates, removing any losing candidate would result in Begich still being the winner.  Thus, the spoiler effect cannot occur with any of our BC methods.  We will demonstrate that truncation and compromise failures are both possible with at least one of the methods, but it is notable that both ABC and MBC avoid all voting failures.

Considering a possible truncation failure, Begich's margin of victory for EBC was only 2738.5 points, and likewise for QBC since they are equivalent when there are only 3 candidates.  If 5478 of the 47407 Peltola$>$Begich$>$Palin voters truncate their ballots to only vote for Peltola without ranking Begich or Palin, then Begich would now only receive 1.5 points instead of 2 points from each of those voters.  This change would result in Peltola now winning the election.  In contrast, ABC would avoid this failure as the margin of victory of nearly 25000 is too large to overcome by any truncation strategy. BCU would also exhibit a truncation failure, although it would require 23037 of 47470 Peltola$>$Begich$>$Palin and Peltola$>$Begich voters to truncate to a single-candidate ballot to trigger this failure.  As noted previously, MBC cannot exhibit a truncation paradox in any election.

A compromise failure can also occur with EBC and QBC in this election.  Some voters who ranked Palin$>$Peltola$>$Begich could decide Peltola has a better chance of winning compared to Palin and ultimately make a compromise by moving Peltola to the first rank. If 1370 of these 3678 voters make this change, Peltola would gain 2 points per voter by moving from 2nd to 1st, overtaking Begich and winning the election.  For ABC, the same set of voters comprising in this way would only gain 1 point each for Peltola.  With such a large margin of victory, there are not enough Palin$>$Peltola$>$Begich to enact a compromise failure.  BCU and MBC avoid a compromise failure by similar reasoning.

We now focus on some elections where the BC methods exhibit a variety of verifiable and manipulation failures, beginning with the \textbf{2022 Alaska House District 6 Election} shown in Table~\ref{Alaska6}.  While also being a 3-candidate election from Alaska, this state legislature election differs from the previously discussed two Republican, one Democrat special election by having a slate of a Republican candidate (Sarah Vance) and two Independent candidates (Ginger Bryant and Louis Flora).  This change in partisan dynamics is likely the cause of verifiable failures of the BC Methods since over 90$\%$ of the voters who listed one of the independent candidates first either left Vance off their ballot or listed her third.

\begin{table}
{\scriptsize
\begin{center}
\begin{tabular}{ | c | c | c | c | c | c | c | c | c | c | }
\hline
 $\#$Votes & 125 & 185 & 41 & 1708 & 2118 & 382 & 3907 & 318 & 735 \\ \hline
 1st & Bryant & Bryant & Bryant & Flora & Flora & Flora & Vance & Vance & Vance \\
 2nd & -- & Flora & Vance & -- & Bryant & Vance & -- & Bryant & Flora \\
 3rd & -- & Vance & Flora & -- & Vance & Bryant & -- & Flora & Bryant \\ \hline
\end{tabular}
\caption{The Alaska State House of Representations District 6 Election in 2022.}
\label{Alaska6}
\end{center}}
\end{table}

The Republican Vance is both the Condorcet winner, as well as the majority winner with 4960, or about $52\%$, of the 9519 first-place votes.  First note that the IRV election declared Vance as the winner, and no voting failures were observed~(Graham-Squire and McCune 2022). Three of the BC variations - EBC, QBC, and BCU - result in Vance winning the election as well; however, ABC and MBC choose the Independent candidate Flora.  Hence these two methods exhibit Condorcet and majority winner failures, the latter vulnerability being one of the primary criticisms of Borda methods.  As we will observe in an upcoming section, this election is a rare example of this failure occurring, and in fact, this is the only election in which more than one variation has a majority winner failure.

Since Vance is a majority candidate with only two other candidates, the spoiler effect and compromise failures are easy to observe or dismiss.  Removing Bryant from the slate would precipitate Vance now becoming the winner of the two-candidate election with ABC and MBC, resulting in a spoiler effect.  Compromise failures, on the other hand, cannot occur with any of the five variations since there are such few Bryant$>$Flora$>$Vance or Bryant$>$Vance$>$Flora voters who could attempt this manipulation.  For instance, the closest margin of victory is Flora winning over Vance using ABC by 92.5 points.  However, there are only 41 voters who could move Vance ahead of Bryant, leaving Vance still trailing Flora by 51.5 points.

The only method that has a truncation failure with this Alaska House race is ABC, which is due to the small 92.5 point advantage for Flora. If 186 of the 735 Vance$>$Flora$>$Bryant truncated their ballot to only Vance, Flora's point total would decrease .5 points per ballot.  Therefore, she would now lose the election to Vance through this voter manipulation failure.

The next election we examine is the \textbf{2021 Minneapolis City Council Ward 2 Election}, which featured five candidates - Tom Anderson, Yusra Arab, Guy Gaskin, Cam Gordon, and Robin Worlobah - where voters were only allowed to vote for up to three candidates in a truncated ballot.  Note with this election, and some of the subsequent examples, we omit the preference profile because of its size. Despite only ranking up to three candidates, voters submitted 84 unique ballots that would need to be displayed.  This election is unique in that it is one of only two elections without a Condorcet candidate (McCune and McCune 2023).  This is due to a Condorcet cycle in which Arab is preferred over Gordon, Gordon over Worlobah, and Worlobah over Arab.  Worlobah ultimately won the election by less 20 votes in the final IRV round over Arab, while also leading in the initial plurality round by less than 50 votes.  Our BC methods all agree on Arab as the winner, but since there is no majority or Condorcet candidate, there is no verifiable failure for either IRV or the BC variations.  IRV was shown~(Graham-Squire and McCune 2022), however, to have demonstrated both upward and downward monotonicity paradoxes, the spoiler effect, and a compromise failure.

The manipulation failures are especially prominent in this election for each BC method.  All four methods that can exhibit a truncation failure (recall that MBC is immune to this) have such a failure.  Additionally, all five BC methods have a compromise failure, and all but BCU have a spoiler effect. To see, for example, how a compromise failure occurs with ABC, we observe that Arab wins with 32799.5 points to Gordon's second place showing of 31850.5, a margin of victory of 949 points.  There are 1589 voters who favor a third candidate while ranking Gordon higher than Arab.  If 950 of them compromise and move Gordon to the top of their ballot, he would gain at least the 950 points needed to defeat Arab, demonstrating a compromise failure.  It is worth noting the difficulty for enacting this failure, as many of these 1589 voters have Gordon as their third choice with Arab unranked, along with having differing top ranked candidates.  Organizing a majority of these voters to compromise in this way would be very challenging if not impossible.

The \textbf{2022 Oakland, CA School Board District 4 Election} is very similar in that it also includes a Condorcet cycle between the three candidates Resnick, Hutchinson, and Manigo.  Four of our BC variations agree on Resnick as the winner of this election, while MBC chooses the winner to be Hutchinson, who also was the IRV election winner. While Resnick has a large lead in first-place votes, as gleaned from Table~\ref{Oakland}, he is also listed third or left off of the most ballots, making this a very close election in terms of total points.  Therefore, akin the Minneapolis election, this school board election features numerous manipulation failures with all five BC methods exhibiting compromise failures and spoiler effects, although only ABC and BCU are vulnerable to a truncation failure.

\begin{table}
{\scriptsize
\begin{center}
\begin{tabular}{ | c | c | c | c | c | c | c | c | c | c | }
\hline
$\#$Votes & 2327 & 2337 & 3563 & 3740 & 3095 & 3180 & 1846 & 4194 & 2150 \\ \hline
 1st & H & H & H & R & R & R & M & M & M \\
 2nd & -- & R & M & -- & H & M & -- & H & R \\
 3rd & -- & M & R & -- & M & H & -- & R & H \\ \hline
\end{tabular}
\caption{The Oakland School Board District 4 Election in 2022 with the candidates Hutchinson (H), Resnick (R), and Manigo (M).}
\label{Oakland}
\end{center}}
\end{table}

The \textbf{2008 Pierce County, WA County Executive Election}, which featured a compromise failure in the IRV election~(Graham-Squire and McCune 2022), is another example in which multiple failures exist for our BC methods.  We first observe in this four-candidate election consisting of Bunney, Goings, Lonergan, and McCarthy, the five Borda variations do not choose the same winner.  ABC, BCU, and MBC agree with IRV as having Pat McCarthy win the election, while EBC and QBC select Shawn Bunney, who we note is the plurality winner.  Since McCarthy is the Condorcet winner, EBC and QBC exhibit a Condorcet winner failure.  With this failure comes a spoiler effect for those two methods, although not for the other three, since removing Goings and Lonergan would result in McCarthy winning head-to-head.

The extremely close nature of this election - McCarthy won the IRV final round 51$\%$ to 49$\%$ - causes several voter manipulation failures. Using ABC, for instance, the candidates McCarthy, Goings, and Bunney finish within 21000 points of one another, which is less than 3$\%$ of the winner McCarthy's 778731.5 points.  This makes truncation and compromise failures easy to occur.  In fact, all five BC methods exhibit a compromise failure and only EBC and MBC avoid a truncation failure.

The \textbf{2009 Burlington, VT Mayor Election} is an election where the IRV result included several unfair aspects, namely a monotonicity paradox, Condorcet winner failure, spoiler effect, and compromise failure~(Graham-Squire and McCune 2022). Of the five candidates, the Condorcet winner was Montroll, while the IRV winner was Kiss.  The five BC variations successfully select Montroll as the winner, avoiding any Condorcet failures as well as spoiler effects.  There are manipulation failures present, including a compromise failure for EBC and QBC, as well as a truncation failure for EBC and BCU.  For instance, when using BCU, Montroll defeats Kiss with 26194 points to 23427.  If 692 of the 1328 voters whose top choice is Kiss followed in second by Montroll truncate their ballot to only list Kiss, then Montroll loses 4 points per ballots.  This causes the winner to change from Montroll to Kiss, resulting in the truncation failure.

The \textbf{2020 San Francisco, CA Board of Supervisors District 7 Election} was one of the two occurrences of a downward monotonicity paradox within the IRV election~(Graham-Squire and McCune 2022).  Our five BC methods agree with the IRV (and Condorcet) winner Melgar, and while they automatically avoid the monotonicity paradox, we again find several manipulation failures. Particularly, EBC and MBC exhibit spoiler effects, all five methods contain a compromise failure, and all but MBC have a truncation failure.  The cause of this large number of failures is twofold: the points earned by the top candidates are very close and with seven candidates on the ballot, the advantages of truncating or compromising are enhanced.  For example, the points given out by QBC for Melgar and second-place Engardio are 445737 and 430048. While the more than 15000 point difference seems large, these are quite close as 17.99$\%$ and 17.36$\%$ of the total points.  Meanwhile a fully-ranked ballot with seven candidates gives out 63 points with 22 to the top candidate and 16 to the second-ranked one, making deficits easy to overcome through swapping choices or removing candidates.

\begin{table}
{\scriptsize
\begin{center}
\begin{tabular}{ | c | c | c | c | c | c | c | c | c | c | }
\hline
$\#$Votes & 35009 & 29717 & 15573 & 33452 & 33546 & 13493 & 11721 & 10854 & 11731 \\ \hline
 1st & C & C & C & R & R & R & V & V & V \\
 2nd & -- & R & V & -- & C & V & -- & C & R \\
 3rd & -- & V & R & -- & V & C & -- & R & C \\ \hline
\end{tabular}
\caption{The 2021 New York City Queens Borough President Democratic Primary with candidates Crowley (C), Richards (R), and Van Bramer (V).}
\label{Queens}
\end{center}}
\end{table}

The election where the BC variatiosn exhibit the largest number of voting failures is the \textbf{2021 New York City Queens Borough President Democratic Primary Election}, seen in Table~\ref{Queens}.  In this three-candidate election with 195096 voters, the top two candidates, Richards and Crowley, were separated by only 192 first-place votes.  The points for these two candidates across the five variations differed by as few as 432 points (344755 to 344323) for BCU and by at most 1597.5 points (502979.5 to 501382) for EBC and QBC. Not only do all five methods have a Condorcet winner failure by selecting Crowley over Richards, but these extremely close results also cause all five methods to have a spoiler effect and compromise failure, with only MBC avoiding a truncation failure.  Recalling that MBC cannot exhibit this failure, we observe that for an election without a majority winner, the maximum number of voting failures occur within this single election. 

Lastly, we examine the \textbf{2018 San Leandro, CA County Council District 3 Election}, which features a rare occurrence that highlights a particular flaw in the application of BCU.  By definition, BCU only awards points to candidates included on a voter's ballot.  Applying this to a two-candidate election is problematic and leads not only to a truncation failure where voters are incentivized to vote only for their preferred candidate, but also majority and Condorcet winner failures.  This seems paradoxical for a two-candidate election, but it is possible with how BCU rewards bullet voting.  As seen in Table~\ref{SanLeandro}, Aguilar receives a majority of first-place votes, 11902 to Thomas's 11358.  However, the single-candidate ballots, which provide the favored candidate 2 points to the unlisted candidate's 0, favor Thomas by a 602 vote difference.  This 1204 point advantage is enough to make Thomas the winner with 30286 points, barely defeating the 30226 for Aguilar.  This outcome is simultaneously a majority winner, majority loser, Condorcet winner, and Condorcet loser failure.

\begin{table}
{\scriptsize
\begin{center}
\begin{tabular}{ | c | c | c | c | c| }
\hline
 $\#$Votes & 4332 & 7570 & 4936 & 6422 \\ \hline
 1st & Aguilar & Aguilar & Thomas & Thomas \\
 2nd & -- & Thomas & -- & Aguilar \\ \hline
\end{tabular}
\caption{The San Leandro, CA County Council District 3 Election in 2018.}
\label{SanLeandro}
\end{center}}
\end{table}

\section{Empirical Results}\label{sect:results}

\begin{table}\label{FailureRate}
{
\begin{center}
\begin{tabular}{ | c | c | c | c | c | c |}
\hline
 Voting Failure & EBC & QBC & ABC & BCU & MBC \\ \hline
 Majority Winner    & 0\% & 0\% & 0.4\% & 1.3\% & 1.3\% \\
 Majority Loser     & 0\% & 0\% & N & 1.8\% & N \\
 Condorcet Winner   & 2.6\% & 3.3\% & 3.1\% & 4.8\% & 5.5\% \\
 Condorcet Loser    & 0\% & 0\% & N & 1.8\% & 0\% \\ \hline
 Upward Monotonicity  & N & N & N & N & N \\
 Downward Monotonicity & N & N & N & N & N \\
 No-Show            & N & N &  N & N & N \\ \hline
 Truncation         & 12.1\% & 17.8\% & 22.3\% & 46.6\% & N \\
 Compromise         & 17.6\% & 16.6\% & 13.8\% & 10\% & 20.4\% \\
 Spoiler Effect     & 4.8\% & 5.1\% & 4.3\% & 3.3\% & 9.6\% \\ \hline 
\end{tabular}
\caption{Summary of voting failure rates for the five BC variations across U.S. RCV elections, where N indicates such a failure is not possible}
\end{center}}
\end{table}

Our analysis of 421 U.S. RCV elections from 2004 to 2023 is displayed in Table~\ref{FailureRate}, which shows the percentage of voting failure occurrences for each BC variation.  Some rates use the complete set of 421 elections, but exceptions include the majority winner and loser failure rates, which are taken out of the 228 elections with a majority winner and the 109 with a majority loser.  Likewise, the Condorcet winner rates use the 419 elections that had a Condorcet winner, and BCU's Condorcet loser rate is based on the 420 elections that had such a candidate (only the Oakland school board election did not).  Finally, the spoiler effect rates use the 396 elections with 10 or fewer candidates.

A primary criticism of the Borda count is its susceptibility to majority winner failures, but note how rare these failures actually are. BCU and MBC's 1.3$\%$ correspond to just 3 elections with majority winner failures, with ABC exhibiting this failure only once.  Meanwhile, EBC and QBC never have this failure, although Example~\ref{MajEBC-QBC} showed one is theoretically possible.  These two methods also avoid majority loser failures and exhibit the lowest rates of Condorcet winner failures, clearly showing the strongest results within the verifiable failure category. BCU, on the other hand, is the only variation with majority and Condorcet loser failures, all of which occurred in two-candidate elections as in the previously described 2018 San Leandro election.

When comparing the manipulation failure rates, the BCU percentage for truncation failures illustrates the obvious incentive voters have with this method to submit ballots that only include their top choice, known as bullet voting.  Nearly half of the elections using this variation are vulnerable to truncation failures, with the three variations featuring the averaged approach for partial ballots performing much better.  This occurs since truncating a ballot in those cases is a zero-sum strategy, meaning that your favored candidate will gain in points compared to other highly ranked candidates, but will instead lose ground over the lowest ranked candidates.  BCU would instead cause your top choice to gain an advantage over all candidates who are removed from the ballot.  The lower truncation failure rate for ABC compared to BCU mirrors the computation results by Kamwa (2022), which were referred to in that paper as the averaged and pessimistic models, respectively.

The extended points-scales of EBC and QBC do feature a disadvantage when it comes to compromise failures compared to ABC and BCU.  The increasing point differences between higher pairs of ranks leaves the method exposed to compromise since exchanging first and second ranked candidates would have an outsized effect on the outcome.  Interestingly, MBC has the highest compromise failure rate, which is likely due to it awarding the lowest amount of total points to the candidates.  This results in a smaller margin of victory that is more likely to be overcome by a compromise strategy.  It also happens to feature a much higher spoiler effect rate than the other four variations.  It should be noted that most of the spoiler effects occur due to Condorcet winner failures since the removal of all candidates other than the original winner and the Condorcet winner would consequently change the outcome.

We cannot directly compare to the failure rates in the analysis by Graham-Squire and McCune (2022) of IRV results because of the use of different sets of election data, but there is still value in observing their differences.  For instance, IRV elections had lower rates for Condorcet winner ($1.1\%$), compromise ($3.8\%$), and spoiler effect failures ($1.6\%$) than any BC variation.  However, IRV's main flaw of antidemocratic voting paradoxes - $1.6\%$ for upward monotonicity, $1.1\%$ for downward monotonicity, and $.5\%$ for no-show - are more prevalent than the majority winner failures for EBC, QBC, and ABC.  

Another observation from our results is that the five BC variations agree on a winner in the vast majority of RCV elections.  In particular, 384 of the 421 elections, or about $91\%$, have the same winning candidate for all five methods.  Furthermore, the three variations EBC, QBC, and ABC that all feature the averaged process for handling partial ballots agree in nearly 96$\%$ of the elections.  The cases where the methods disagree are where many of the voting failures arise, but it is notable that the choice of variation is irrelevant for such a high proportion of elections.

\section{Conclusion}

Our analysis of Borda count variations highlights both their strengths and limitations.  Not only are the methods designed to avoid antidemocratic voting paradoxes, but they also perform quite well in terms of verifiable failures.  Notably, majority winner failures are rare, and EBC in particular avoids both majority winner and loser failures while minimizing the Condorcet winner failures.  

The clear trade-off with the five variations compared to IRV lies in the increased vulnerability to voter and candidate manipulation failures.  The averaged approach to partial ballots does at least improve the issue of truncation failures, mitigating the risk of widespread bullet voting.  While truncation and compromise failures are certainly more common with the BC variations than with IRV, this should not preclude the implementation of these methods in additional election jurisdictions.  Coordination among different voters groups required to enact a common truncation or compromise strategy remains challenging, limiting the practical impact of these vulnerabilities.

The primary flaws of IRV, particularly monotonicity and no-show paradoxes, cannot easily be fixed. In contrast, BC methods can be implemented in a way to completely avoid majority failures if that is a major concern for a voting jurisdiction. Black's method, for instance, would always select the Condorcet winner and defaults to a Borda count in elections without one, which would remove all verifiable failures while reducing the manipulation failures.  Ranked pairs and Schulze methods are two other methods that have the same desirable outcome, although their complexity would lack the transparency and accountability of a point-based election method.  These other monotonic voting methods warrant further theoretical and empirical study.

More research is also needed for these BC variations and for RCV methods in general on their outcomes with different types of elections.  For example, do certain methods exhibit more voting failures when the number of candidates is higher, or within nonpartisan elections?  Jungle primaries, which often lead to general elections with multiple candidates from the same party, could also lead to a higher risk of unfair outcomes. A deeper understanding of these dynamics could inform the design and adoption of monotonic voting methods such as variations of the Borda count.







\end{document}